%% file: arxiv_version.tex
\documentclass [12pt]{article}

\usepackage{placeins}
\usepackage{amsmath,amsthm,amscd,amssymb}
\usepackage[linesnumbered]{algorithm2e}
\usepackage{bm}
\bmdefine{\bz}{z}
\bmdefine{\bTheta}{\Theta}

\usepackage[small, centerlast, it]{caption}
\usepackage{epsfig}
\usepackage[titles]{tocloft}
\usepackage[latin1]{inputenc}
\usepackage{verbatim} 
\usepackage[active]{srcltx} 
\usepackage{fancyhdr}
\usepackage{caption}

\def\b1{{1\!\!1}}

\def\cA{{\ca A}}

\def\sH{{\mathsf H}}

\def\bC{{\mathbb C}}           

\def\bR{{\mathbb R}}
\def\bP{{\mathbb P}}

\def\mM{\mathcal M}

\def\vz{{\bf z}}
\def\vt{{\bf \Theta}}

\def\beq{\begin{eqnarray}}
\def\eeq{\end{eqnarray}}

\newcommand{\ca}[1]{{\cal #1}}         

\newcommand{\ket}[1]{|{#1}\rangle}

\usepackage{sectsty}
\sectionfont{\normalsize}
\subsectionfont{\normalsize\normalfont\itshape}
\newtheoremstyle{thm}
{12pt}
{12pt}
{\itshape}
{}
{\itshape\bfseries}
{}
{1em}
{}

\theoremstyle{thm}
\newtheorem{theorem}{Theorem}

\newtheorem{proposition}[theorem]{Proposition}
\newtheorem{definition}[theorem]{Definition}
\newtheorem{ex}[theorem]{Example}

\def\hs{\hspace{0.3cm}}

\begin{document}

\hfill{\sl  } 
\par 
\bigskip 
\par 
\rm


\vspace{-2cm}
\begin{center}
\par
\bigskip
\large
\noindent
{\bf Quantum Annealing Learning Search \\for solving QUBO problems}
\bigskip
\par
\rm
\normalsize 

\date{\today}

\noindent  {\bf Davide Pastorello$^\star$} and 
\noindent{\bf Enrico Blanzieri$^\dagger$\footnote{The authors contributed equally.}}

\vspace{0.5cm}

\noindent
$^\star$  \emph{Department of  Mathematics, University of Trento \\ 
$ $\hs Trento Institute for Fundamental Physics and Applications}\\ 
$ $\hs via Sommarive 14, 38123 Povo (Trento), Italy\\
$ $\hs d.pastorello@unitn.it

\vspace{0.5cm}

\noindent
$^\dagger$ \emph{Department of Engineering and Computer Science, University of Trento}\\
$ $\hs via Sommarive 9, 38123 Povo (Trento), Italy\\
$ $\hs enrico.blanzieri@unitn.it 
\smallskip
\end{center}
 \normalsize

\par

\rm\normalsize


\rm\normalsize


\abstract{\noindent In this paper we present a novel strategy to solve optimization problems within a hybrid quantum-classical scheme based on quantum annealing, with a particular focus on QUBO problems. The proposed algorithm is based on an iterative structure where  the representation of an objective function into the annealer architecture is learned and already visited solutions are penalized by a tabu-inspired search. The result is a heuristic search equipped with a learning mechanism to improve the en- coding of the problem into the quantum architecture. We prove the convergence of the algorithm to a global optimum in the case of general QUBO problems. Our technique is an alternative to the direct reduction of a given optimization problem into the sparse annealer graph.}

\section{Introduction}

Quantum Annealing (QA) is a type of heuristic search to solve optimization problems \cite{r1, r2, aqc}. The solution to a given problem corresponds to the ground state of a quantum system with total energy described by a \emph{problem Hamiltonian} $H_P$ that is a self-adjoint operator on the Hilbert space where the considered quantum system is described. The main idea is to set a time evolution in the interval $[0,T]$ given by a time-dependent Hamiltonian $\{H(t)\}_{t\in[0,T]}$ such that $H(T)=H_P$ providing that the quantum system ends in the ground state of $H_P$ with high probability. In QA the time-dependent Hamiltonian of the considered quantum system (typically a $n$-qubit system) is:
\beq
H(t)=H_P+\gamma(t) H_D,
\eeq 
where $\gamma:[0,T]\rightarrow\bR$ is a smooth monotone function decreasing to 0. $H_D$ is the so-called \emph{disorder Hamiltonian} which does not commute with $H_P$, and it introduces the quantum fluctuations to allow the heuristic search escaping local minima \cite{aqc}. The annealing procedure is implemented decreasing the kinetic contribution of $H_D$. 
The difference between QA and Simulated Annealing (SA) is that QA exploits the \emph{tunnel effect} to escape local minima instead of \emph{thermal hill-climbing}.  
\\
QA can be physically realized considering a quantum spin glass that is a network of qubits arranged on the vertices of a graph whose edges represent the interactions between neighbors. In particular the D-Wave quantum annealer implements a graph with the \emph{Chimera topology} (or the more recent \emph{Pegasus topology}) \cite{headquarters2018d,d-wave} and a general problem Hamiltonian $H_\vt$, as defined below in (\ref{HP}), while a transverse external field implements the disorder Hamiltonian $H_D$ as defined in (\ref{HD}). 
\\
The general optimization problem that the D-Wave's machine solves is expressed in terms of the \emph{Ising model}, i.e. its solution ${\bf z}_0\in\{-1,1\}^n$ is given by the minimization of the following cost function:
\beq\label{ising}
\mathsf E(\vt, \vz)=\sum_{i\in V} \theta _i z_i +\sum_{(i,j)\in E} \theta_{ij}z_i z_j\hs\hs\vz=(z_1,...,z_n)\in\{-1,1\}^n
\eeq
where $(V,E)$, with $|V|=n$, is a subgraph of the architecture employing Chimera (Pegasus) topology, the binary variables $z_i\in\{-1,1\}$ are physically realized by the outcomes of measurements on the qubits located in the vertices and $\vt$ is the matrix of the coefficients of $\mathsf E$, called \emph{weights}, defined as:
\beq \label{w}
\vt_{ij}:=\left\{
\begin{array}{cc}
\theta_i & i=j\\
\theta_{ij} & \hs(i,j)\in E\\
0 & \hs(i,j)\not\in E .
\end{array}\right.
\eeq
The cost function (\ref{ising}) is realized by means of the Hamiltonian of a spin glass:
\beq\label{HP}
H_\vt=\sum_{i\in V} \theta _i \sigma_z^{(i)} +\sum_{(i,j)\in E} \theta_{ij}\sigma_z^{(i)} \sigma_z^{(j)}.
\eeq

\noindent
$H_\vt$ is an operator on the Hilbert space $\sH=(\bC^2)^{\otimes n}$ with $\sigma_z^{(i)}:=I_2\otimes\cdots \otimes \sigma_z\otimes\cdots\otimes I_2$ where $I_2$ is the $2\times 2$ identity matrix and the Pauli matrix
$$
\sigma_z=\left(
\begin{array}{cc}
1 & \,\,\,\,0\\
0 & -1
\end{array}
\right)
$$
is the $i$-th tensor factor.
\\
A transverse field gives the kinetic Hamiltonian to realize the quantum fluctuations:
\beq\label{HD}
H_D=\sum_{i\in V} \sigma_x^{(i)},
\eeq 
where we adopt the same notation of (\ref{HP}) and $\sigma_x$ is the Pauli matrix:
 $$
\sigma_x=\left(
\begin{array}{cc}
0 & \,1\\
1 & 0
\end{array}
\right).
$$
\\
Given a problem with a binary encoding, the annealer is initialized by a suitable choice of the weights $\vt$ and a definition of a mapping of the binary variables in the corresponding qubits.
The satisfying assignment of the considered problem corresponds to the ground state $\ket{\vz_0}$ of the system whose energy eigenvalue is: 
\beq\label{argmin_zE}
\vz_0=\mbox{arg}\!\!\!\!\!\!\!\!\!\!\!\min_{\vz\in\{-1,1\}^{|V|}\,\,\,\,\,\,} \!\!\!\!\!\!\!\mathsf E(\vt,\vz).
\eeq
If the problem presents distinct solutions then $H_P$ has a degenerate ground state such that the quantum annealing process and a subsequent measurement produce one of the possible solutions. In the case of the real machine we have also to take into account the repetitions of the annealing and measurement processes. In fact the annealer returns a sample of solutions which may or not may contain an optimal one, so we also consider the behavior of quantum annealer as an estimation of $\vz_0$:
\beq\label{hatargmin_zE}
\hat\vz_0=\!\!\!\!\!\!\!\widehat{\,\,\,\,\,\,\,\,\,\mbox{arg}\!\!\!\!\!\!\!\!\!\!\!\min_{\vz\in\{-1,1\}^{|V|}\,\,\,\,\,\,}} \!\!\!\!\!\!\!\mathsf E(\vt,\vz)
\eeq
where the right-hand term represents a process that includes the runs of the annealer and the computation of an estimate from their results. For our purpose this estimation could be either the selection of one of the values or a statistic on the annealer outputs, for example the mode. Details about QA, its realization by means of spin glasses and its application to solve optimization problems can be found in literature \cite{r1,r2,r3}. 
\\
In this work we propose a procedure based on iterated modifications of the Hamiltonian (\ref{HP}) to implement a \emph{tabu-inspired search}  to solve optimization problems that cannot be directly mapped in the architecture of a quantum annealer. In particular we will focus on the \emph{quadratic unconstrained binary optimization} (QUBO) problem that is defined by:
\beq
\mbox{Minimize the function} \hs f(\vz)=\vz^T Q \vz \hs\hs \vz\in\{-1,1\}^n, 
\eeq
where $Q$ is a $n\times n$ symmetric matrix of real coefficients. The qubit network of a quantum annealer can represent a QUBO problem at the price of restricting the QUBO structure into a graph with less vertices and edges (imposed by the hardware architecture), such a mapping is the subject of intense research activity (e.g. \cite{Roberto}). Our approach is different: Implementing an iterative representation of the objective function  and a tabu search, we can solve a QUBO problem by means of an original utilization of a quantum annealer without the direct representation of the problem into the annealer graph.   
\\
Let us recall that \emph{tabu search} is a kind of local search where worse candidate solutions can be accepted with nonzero probability and already visited solutions are penalized \cite{g1, g2}. 
In the next section we describe our general procedure to implement a tabu-inspired search iterating the initialization of a quantum annealer. 
The considered scheme produces the candidate solutions for the evaluation of the objective function alternating annealing and re-initialization of the weights in order to penalize already-visited solutions. There is no a list of forbidden moves like in standard tabu search but rather the implementation of a prescription to discourage certain moves and to encourage others towards a better representation of the problem into the quantum architecture. 
In Section \ref{QCTS} we propose the algorithms: Algorithm 1 is a general scheme of learning search on a quantum annealer and Algorithm 2 is a specific version to solve QUBO problems for which we give a convergence proof (Section \ref{Conv}).

\section{Learning search on a quantum annealer}

In order to introduce the main idea of a learning strategy for local search in the context of quantum annealing, let us consider a simple example of two qubits where $(V,E)=(\{v_1, v_2\}, \{e\})$ and the cost function (energy) is given by:
\beq
\mathsf E(\vt, \vz)=\theta_1 z_1+\theta_2 z_2+ \theta_{12} z_1 z_2.
\eeq
Moreover, let us also consider an objective function $f:\{-1,1\}^2\rightarrow\bR$ to minimize (not necessary a QUBO problem for now). The heuristic search starts with the random initialization of the weights $\vt_1$ and $\vt_2$, then two annealing processes produce the first two candidate solutions: $\vz_g$ and $\vz_b$ with $f(\vz_g)\leq f(\vz_b)$, i.e. $\vz_g$ is the current solution. The basic idea is to generate  a set of new weights $\vt'$ for a subsequent run of the quantum annealer such that the \emph{bad candidate} $\vz_b=(z_{b1},z_{b2},z_{b3})$, i.e. the already discarded solution, is energetically penalized:
\beq\label{w1}\left\{\begin{array}{c}
\!\!\!\!\!\!\!\!\!\!\theta_1'=z_{b1}\\ \!\!\!\!\!\!\!\!\!\!\theta'_2=z_{b2}\\ \theta'_{12}=z_{b1}z_{b2} .
\end{array}\right. 
\eeq
Once the weights (\ref{w1}) are initialized a new candidate solution is produced. Let us assume that the corresponding value of $f$ is different from $f(\vz_g)$ then we have two bad candidates $\vz_b^{(1)}$ and $\vz_b^{(2)}$. In this case the weight initialization for the next run of the annealer will be:
\beq\label{w2}\left\{\begin{array}{c}
\!\!\!\!\!\!\!\!\!\!\!\!\!\!\!\!\!\theta_1''=z_{b1}^{(1)}+z_{b1}^{(2)}\\ \!\!\!\!\!\!\!\!\!\!\!\!\!\!\!\!\!\!\theta''_2=z_{b2}^{(1)}+z_{b2}^{(2)}\\ \,\,\theta''_{12}=z_{b1}^{(1)}z_{b2}^{(1)}+z_{b1}^{(2)}z_{b2}^{(2)}.
\end{array}\right.
\eeq
\begin{ex}[Toy model for two qubits]
After the two initial annealing processes suppose the outputs be  $\vz_g=(1,1)$ and $\vz_b=(1,-1)$ then $\mathsf E(\vt', \vz)=z_1-z_2-z_1z_2$ with:
$$
\emph{arg}\!\!\!\!\!\!\!\!\!\min_{\vz\in\{-1,1\}^{2}\,\,\,\,\,\,} \!\!\!\!\!\!\!\mathsf E(\vt',\vz)=\{(1,1), (-1,1), (-1,-1)\}
$$
then assume to perform an annealing process producing the result $(-1,1)$ satisfying $f(-1,1)<f(1,1)$, so we have two bad candidates: $\vz_b^{(1)}=(1,-1)$ and $\vz_b^{(2)}=(1,1)$ and a new current solution $\vz_g=(-1,1)$. The weight initialization for the next annealing process will be:
$$\left\{\begin{array}{c}
\!\theta_1''=2\\ \!\theta''_2=0\\ \,\,\theta''_{12}=0
\end{array}\right.
$$
so $\mathsf E(\vt'', \vz)=z_1$ for which $\emph{arg}\!\min \mathsf E(\vt'',\vz)=\{(-1,1), (-1,-1)\}$. If $\vz^*=(-1,-1)$ is the global optimum for the objective function $f$ then the last step of this toy learning search is the initialization of the cost function:
$$\mathsf E(\vt''', \vz)=z_1+z_2-z_1z_2$$
attaining the global minimum in $\vz^*$.
\end{ex}

\noindent
Since the function $\mathsf E$ does not encode any information about $f$ the search sketched in the example above is obviously equivalent to an exhaustive search at this stage. \\Let us give a first generalization of the argument above considering $n$ qubits. Suppose to have collected a list $\{\vz^{(\alpha)}  \}_{\alpha\in I}\subset \{-1,1\}^n$ of discarded candidates, i.e. there is a current solution $\vz_g$ such that $f(\vz^{(\alpha)})\geq f(\vz_g)$ for any $\alpha\in I$. These elements can be penalized by the following choice of the weights as a direct generalization of (\ref{w1}) and (\ref{w2}):
\beq\label{tabu1}
\left\{\begin{array}{cc}
\!\!\!\!\!\!\theta_i=B\sum_{\alpha\in I} z_i^{(\alpha)} &\\ 
 & \\
 \,\,\theta_{ij}=C\sum_{\alpha\in I} z_i^{(\alpha)} z_j^{(\alpha)} &  (i,j)\in E
\end{array}\right.
\eeq
\noindent
where $B$ and $C$ are positive constants. \\ Let us define the \emph{tabu matrix} as the $n\times n$ symmetric matrix with integer elements  given by (\ref{tabu1}):
\beq\label{tabumatrix}
S:=\sum_{\alpha\in I}[ \vz^{(\alpha)}\otimes\vz^{(\alpha)}-I_n +\mbox{diag}(\vz^{(\alpha)})].
\eeq
A tabu procedure updates the matrix $S$ when the number of rejected candidates increases  (details in Section \ref{QCTS}), so we consider its equivalent definition by recursion:
\beq\label{tabumatrix1}
S:=S+\vz\otimes\vz-I_n +\mbox{diag}(\vz),
\eeq
with $S$ initialized as the null matrix and $\vz$ ranging in the list of bad candidates.\\
In order to encode information of the objective function $f$ into the weights of the energy function $\mathsf E$ we introduce a map $\mu:\mathcal O\rightarrow\mathcal E$ from the set of objective functions $\mathcal O:=\{f:\{-1,1\}^n\rightarrow\bR\}$ to the set of energy functions:
\beq
\mathcal E:=\{\mathsf E(\vt, \cdot) \,|\, \vt\in\mM_E \},
\eeq
where $\mM_E$ denotes the set of matrices of the form (\ref{w}).
The general form of $\mu$ that we consider is the following:
\beq\label{tab}
\mu[f](\vz):=\mathsf E(\vt[f], \pi(\vz)),
\eeq
where $\vt[f]$ is the matrix of weights associated to $f$ according to a certain law and $\pi$ is a permutation of the variables changing the binary encoding into the qubits of the annealer. 
\\
We are interested in the implementation of the learning search by means of the tabu matrix $S$ (\ref{tabumatrix}) which encodes the energetic penalization of the rejected candidate solutions. To this end we need to clarify how the matrix $S$ and the permutation $\pi$ interact.
Let $S$ be the tabu matrix generated by the bad candidates $\{\vz^{(\alpha)}\}_{\alpha\in I}$ according to (\ref{tabumatrix}) and $S_\pi$ be the tabu matrix generated by $\{\pi(\vz^{(\alpha)})\}_{\alpha\in I}$, if $P_\pi$ is the matrix of the permutation $\pi$ then $S_\pi=P_\pi^T SP_\pi$. So the tabu-implementing encoding of $f$ into the cost function can be defined by:
\beq\label{tabumu}
\mu[f](\vz):=\mathsf E(\vt[f]+\lambda\, S_\pi\circ \cA, \pi(\vz)),
\eeq
where $\cA$ is the adjacency matrix of the annealer graph, the Hadamard product $\circ$ maps $S_\pi$ into a matrix of $\mM_E$, and the factor $\lambda>0$ balances the contribution of the tabu matrix and the weights $\vt[f]$ that carry information about $f$. 
 The basic idea of definition (\ref{tabumu}) is that the additive contribution of the tabu matrix energetically penalizes the already-rejected candidate solutions.
The goal of $\lambda$ is to avoid that the action of the tabu matrix $S$ obscures the information about $f$, in general $\lambda$ is a decreasing function of the number of bad candidates penalized by $S$.
\\
Once the quantum annealing process has produced an estimate of the global minimum of (\ref{tabumu}), one must apply $\pi^{-1}$ on the result to read the solution in the original variable encoding. During the iterative search a sequence of $\mu$-mappings is generated to explore the structure of the objective function. We use a permutation of the variables in (\ref{tab}) and (\ref{tabumu}) for sake of simplicity but any invertible transformation on $\{-1,1\}^n$ is eligible. This class includes the permutations composed with local changes of sign, however it is not clear which advantage these further transformations could give within the considered scheme.
\\
In the specific case of QUBO problems, the objective function is a quadratic form $f:\{-1,1\}^n\rightarrow \bR$ with $f(\vz)=\vz^T Q \vz$ where $Q$ is a real symmetric $n\times n$ matrix.  
The tabu strategy and the evolving representation of the problem in the annealer can be summarized as follows: The matrix $Q$ representing the objective quadratic form is \emph{piecewise} mapped into the annealer architecture and \emph{deformed} by means of the tabu matrix to energetically penalize the spin configurations corresponding to solutions that are far from the optimum. Note that in general the current representation of the problem is \emph{not} a subproblem. More precisely let us define the mapping $\mu$ in the QUBO case:
\beq\label{encod0}
\mu[f_Q]:= \mathsf E(P^T Q P\circ \cA ,P\vz),
\eeq 
where $P$ is a permutation matrix of order $n$ and $P^T$ is its transpose. Thus the action of $\mu$ is realized by mapping some elements of $Q$ into the weights, the mapped elements are selected by $P$.
The tabu-implementing encoding of $Q$ into the annealer induced by the permutation matrix $P$ turns out to be:
\beq\label{encod}
\,\mu[f_Q](\vz)=\mathsf E( P^T Q P \circ \cA + \lambda\,P^T S P\circ \cA, P\vz)=
\eeq 
$$\hs=\mathsf E( P^T (Q+\lambda S) P \circ \cA, P\vz).\hs\hs\hs$$

\noindent
In the next section, the form (\ref{tabumu}) of the tabu-implementing encoding is included in a general scheme for quantum annealing learning search, whose QUBO version uses (\ref{encod}) instead.

\section{Hybrid quantum-classical learning search}\label{QCTS}

\input{algorithm_description.tex}

\section{Convergence }\label{Conv}

Algorithm~\ref{Alg} has characteristics, including its simulated annealing structure, that permit to reach conclusions about its convergence. In order to prove that it converges to a global optimum of the considered QUBO problem we basically apply a result \cite{tabu} for simulated annealing modeled in terms of Markov processes.
\\
Let us summarize the general result which implies the convergence of our algorithm. Consider a simulated annealing process to find the global minimum of an objective function $F:X\rightarrow \bR$ defined on a finite set $X=\{x_1,...,x_n\}$. Let $A$ be a stochastic $(n\times n)$-matrix\footnote{Namely the elements $\{a_{ij}\}$ of $A$ satisfy: $a_{ij}\geq 0$ $\forall i,j$ and $\sum_{j} a_{ij}=1$ $\forall i$.}, called \emph{generation matrix}, defined as follows: If $x_i$ is the current solution then the element $a_{ij}$ is the probability that $x_j$ is found as a candidate for the next solution. Once a candidate solution is selected an acceptance probability must be defined. Let us assume that if $x_j$ is a candidate to be the successor of $x_i$ then it is accepted as the new current solution with probability:
\beq\label{accept}
\bP_{ij}(T)=\left\{
\begin{array}{cc}
1 &  F(x_j)<F(x_i)\\
e^{-\frac{F(x_j)-F(x_i)}{T}} & \mbox{otherwise}
\end{array}
\right.
\eeq
where $T$ is the decreasing temperature parameter of the simulated annealing. Definition (\ref{accept}) entails that a suboptimal solution has nonzero probability to be accepted and it vanishes as $T\rightarrow 0$.
\\
For $A$ and $\bP_{ij}$ given, the single steps of the process are described by the transition matrix $M(T)$ whose elements
\beq\label{M}
m_{ij}(T):=a_{ij}\cdot \bP_{ij}(T)
\eeq
are the probability that if the current solution is $x_i$ then the next solution will be $x_j$ at temperature $T$. Modeling a simulated annealing process as a Markov chain we can state the following definition.
\begin{definition}
Let $M(T)$ be the transition matrix at temperature $T$ of a simulated annealing process with objective function $F:X\rightarrow\bR$. The probability distribution $\pi_T$ on $X$ satisfying 
\beq
\pi_T(x_j)=\sum_{i=1}^n \pi_T(x_i)m_{ij}(T)\eeq
 is called \textbf{stationary distribution} of $M(T)$.
\end{definition}

\noindent
Assuming the transition matrix $M(T)$ to be irreducible in the sense of Markov chains, i.e. $x_i$ is accessible from $x_j$ for all $i,j$, and it is aperiodic for any $T>0$ then $M(T)$ has a unique stationary distribution $\pi_{T}$ for any $T>0$ and $\lim_{T\rightarrow 0}\pi_T$ exists \cite{Ha02}.\\
 Following \cite{tabu} a simulated annealing scheme can be explicitly realized by means of a number $L$ of decreasing temperature levels $T_1\geq \cdots\geq T_L\simeq 0$ and a number $N$ of iterations at each temperature level. So the statistic produced by many iterations at the temperature level $T_l$ reproduces the distribution $\pi_{T_l}$.
Now let us consider a remarkable result about Markov chains \cite{af1, af2} in the following proposition.
\begin{proposition}
Let $\{M(T)\}_{T>0}$ be a family of transition matrices and $\{T_l\}_{m=1,...,L}$ a set of temperature levels such that $T_1\geq \cdots\geq T_L\simeq 0$.  \\
If the inhomogeneous Markov chain defined by $\{M(T_l)\}_{l=1,...,L}$ converges then its limit distribution corresponds to 
\beq\label{stationary}
\pi^*=\lim_{T\rightarrow 0} \pi_T,
\eeq
where $\pi_T$ is the stationary distribution of $M(T)$ for any $T>0$.
\end{proposition}

\noindent
As a consequence the annealing process modeled by the chain is successful, i.e. there is the convergence to one of the global optima, when $\{M(T_l)\}_{l=1,...,L}$ converges and $\pi^*$ is nonzero only on the global minima of the objective function. About the convergence of the inhomogeneous Markov chain we have that it converges if the temperature descent is sufficiently slow \cite{af1, af2}, i.e. $N\cdot L\rightarrow +\infty$ (a good choice is considered to be $L\simeq 10$ and $N\simeq 10^3$ \cite{tabu}).\\
In order to check whether $\pi^*$ of the considered simulated annealing process recognizes the global optima we must consider the \emph{neighborhood graph}, i.e. the directed graph $G=(X,\mathfrak E)$ defined by $X$ as the vertex set and $\mathfrak E:=\{(x_i,x_j): a_{ij}>0\}$ where $A=\{a_{ij}\}$ is the generation matrix defined above. The assumption that $M(T)$ is irreducible for every $T>0$ is equivalent to the fact that $G$ is strongly connected\footnote{$G$ is strongly connected if for any pair of vertices $x_i$ and $x_j$ there is a direct path connecting them.}.    

\begin{definition}
The neighborhood graph $G=(X,\mathfrak E)$ is said to be \textbf{weakly reversible} when for every $\tilde F\in\bR$ and any $x_i,x_j\in X$ we have that $x_j$ can be reached from $x_i$ along a direct path $(x_i, y_1,...,y_m,x_j)$ such that $F(y_k)\leq \tilde F$ $\forall k=1,...,m$ if and only if $x_i$ can be reached form $x_j$ along such a path.
\end{definition}

\noindent
Now let us consider a more general model where the generation matrix may be varying with the temperature parameter, indeed we have a family of generation matrices $\{A(T)\}_{T>0}$. This is a remarkable feature for the tabu search applications \cite{tabu} where the generation matrix does not depend only on $T$ but also on the information gained in previous iterations. 
In \cite{tabu0} there is a crucial result to establish the convergence of a generalized simulated annealing algorithm characterized by a family $\{A(T)\}_{T>0}$ of generation matrices.
\begin{theorem}\label{main}
Let $\{A(T)\}_{T>0}$ be a family of generation matrices for a simulated annealing process modeled as a Markov chain within the above scheme.\\
If:\\
\\
i) There exists a $\delta>0$ such that $a_{ij}(T)>0$ implies $a_{ij}(T)\geq \delta$ for $i\not =j$ and $\forall T>0$;
\\
ii) The neighborhood graph $G$ of $A(T)$ does not depend on $T$;
\\
iii) $G$ is weakly reversible;
\\
\\
then $\pi^*(x)>0$ if and only if $x\in X$ is an optimal solution. 
\end{theorem}

\noindent
We can apply Theorem \ref{main} to prove the convergence of Algorithm \ref{Alg}, in fact our iterative scheme can be considered a simulated annealing process where the role of temperature parameter is played by the probability $p$ by means of the relation $T=-\ln^{-1}p$. Moreover the search realized by Algorithm \ref{Alg} presents the structure of temperature levels with many runs at each level, in fact the parameter $p$ is constant for $N$ iterations. These cycles provide the convergence of the iterative search modeled as an inhomogeneous Markov chain.

\begin{proposition}\label{convergenza}
Algorithm~\ref{Alg} converges to one of the global minima of $x^TQx$.
\end{proposition}
\begin{proof} 
To prove the statement we show that 1) the proposed algorithm presents the structure described above and 2) it verifies the hypotheses of Theorem \ref{main}. \\
1) First of all let us recall that the generation of a new candidate solution, with a corresponding generation matrix, is described from line 19 to line 27. Moreover let us  check that the acceptation probability is of the form (\ref{accept}), i.e. if $\vz^*$ is the current solution then the probability that $\vz'$ is accepted as the new current solution  has the following form:
\beq\label{ap}
\bP(\vz^*,\vz',T) :=\left\{\begin{array}{cc}
1 & f'<f^*\\
e^{-\frac{f'-f^*}{T}} & \mbox{otherwise}
\end{array}\right. 
 \eeq
where $f'=f(\vz')$ and $f^*=f(\vz^*)$. The acceptation probability of our search corresponds to (\ref{ap}) by posing $T=-{\ln}^{-1} p$ as shown from line 30 to line 36 of the algorithm.  The temperature lowering of this simulated annealing process is described by line 21.\\
For a fixed value of $p$ we can identify a single step of the iterative search from line 19 to line 41 that can be described by a transition matrix $M(p)$. The \emph{repeat-until}, and the \emph{if} statement of lines 20-22 realize the structure of temperature levels, so the corresponding inhomogeneous Markov chain converges as $i_{max}\rightarrow +\infty$. 
\\
2) Let us initially focus on hypothesis \emph{iii)} of Theorem \ref{main}. It suffices to show that if a solution was reached from another solution then it is also possible to do the contrary because this implies weak reversibility. By denoting with $s$ the number of components that are equal between the two solutions, the probability to generate again the previous solution is $q\cdot (1-p)^{s} \cdot p^{n-s}$ as an effect of line 27 and Algorithm~\ref{Alg4}. Since $p_\delta<0.5$, we have that $q\cdot (1-p)^{s} \cdot p^{n-s}\geq q\cdot (\eta p_\delta)^n$.  
 This last observation leads to see that also hypothesis \emph{i)} holds, in fact $q\cdot (\eta p_\delta)^n$ corresponds to $\delta>0$ of Theorem \ref{main}, for $q$, $\eta$ and $p_\delta$ are strictly positive. 
\\
The above observation holds indeed for any pair of solutions. Consequently hypothesis \emph{ii)} is also satisfied because the neighborhood graph associated to the generation is complete for any $p$, in fact from any current solution $\vz^*$ there is a probability larger than $q\cdot (\eta p_\delta)^n$ to reach any other solution.  
\end{proof}

\noindent
Let us make apparent how an aspect of the proposed Quantum Annealing Learning Search does not have any consequence on its convergence. Describing Algorithm 1 and Algorithm 2 we have not assumed that the quantum annealer always returns the optimal solution to the annealer energy function (\ref{argmin_zE}) in favor of the realistic assumption that the annealer returns a sample of solutions which may or not may contain an optimal one (\ref{hatargmin_zE}). However, even if the estimation of the ground state solution in (\ref{hatargmin_zE}) and used in lines 8 ad 26 of Algorithm~\ref{Alg} was perfect, namely it coincided to (\ref{argmin_zE}), the algorithm converges anyway. In fact there is always a non-null probability that a non-optimal solution is accepted (Algorithm~\ref{Alg}, line 35).
On the other hand, the algorithm converges also when the quantum machine returns a poor estimation of a global minimum of the annealer energy function most of the time. In this case the procedure reduces to a classical simulated annealing with random generation of the candidate solutions whose convergence is guaranteed \cite{tabu}.
\\
Note  that Proposition \ref{convergenza} holds for any QUBO problem, including also those whose solution does not contain solutions of subproblems, namely a violation of the so-called \emph{optimal-substructure property}. 
 In fact our algorithm does not decompose the problem in subproblems but solves a series of problems defined by the additive interaction between the weights that carry information about $f$ and the tabu matrix. In general this series is not composed by subproblems of the full problem consistently with the possible absence of the optimal-substructure property.  \\
A crucial aspect for convergence of Algorithm~\ref{Alg} is the probabilistic behavior of the scheme. Let us point out where quantum and classical probabilities intervene in the iterative part of the algorithm: 1) Quantum probability in the output of the annealer (Algorithm~\ref{Alg}, lines 25 and 26); 2) Classical probability $p$ of permutation modification (line 23); 3) Classical probability $q$ of solution perturbation (line 27); 4) Classical probability $p^{f'-f^*}$ of acceptance of suboptimal solutions (line 35). The probabilities 1-3) characterize the process of solution generation whereas 4) influences the solution acceptance probability. 
As shown in the proof of Proposition \ref{convergenza}, classical probabilities $p$ and $q$ allow for a nonzero probability to reach any candidate solution, this fact is the keystone of the convergence argument. 

\noindent
The proof of Proposition \ref{convergenza} makes evident that our algorithm presents the structure of a simulated annealing that delegates part of the search to a quantum annealer. However the quantum annealer does not search purely on a representation of the objective function. In fact we stress that the tabu-implementing encoding produces its effects in the quantum part of the scheme. Thus the presented algorithm cannot be reduced to a simple wrapping of quantum annealer calls in a simulated annealing procedure. This implies that the power of the quantum annealing acts on a mixed representation of the objective function and of the current state of the search that includes the bad elements arranged in the tabu matrix. This circumstance may give computational advantages that must be further analytically and empirically investigated.  
\\
Let us clarify that the given proof ensures the convergence asymptotically as the number of iterations goes to infinity. Empirical evaluation is needed to elucidate the actual number of iterations required in practice for each class of problems. 
It is interesting to note that Proposition \ref{convergenza} also shows a constructive way of finding values of parameters $\bTheta$ that encode the very same minima of the QUBO problem. This means that it would be possible for a class of QUBO problems to run several times the algorithm on different instances and learn the mapping into the annealer architecture from examples. In this way the learned mapping could be used to initialize other instances of the same class of problems decreasing the required number of iterations.

\input{related.tex}

\section{Conclusion and future work}

We have presented a general scheme for the application of a nonstandard variation of tabu search to define a learning search, a specific instance of a hybrid quantum-classical algorithm, and proved its convergence. \\
The general strategy of the proposed algorithm is to implement a search by  defining and updating a tabu matrix that deforms the weights of the annealer energy function. Moreover the representation of the objective function into the annealer graph evolves during the search depending on the current solution and the tabu matrix. In Algorithm~\ref{Alg}, the evolution of the encoding is determined by the generation of permutations depending on a probability that can be interpreted as the decreasing temperature of a simulated annealing process. Applying some results about convergence of SA processes modeled by Markov chains, we proved the convergence of Algorithm~\ref{Alg} to a solution of the QUBO problem.\\
The main direction for further investigation is the estimation of the run time of the presented algorithm (Algorithm~\ref{Alg}) and its empirical validation. Let us point out that if we choose to construct $10$ temperature levels with $10^3$ iterations for each level, we need to call the annealer $k\cdot10^4$ times, where $k$ is the number of annealer calls per iteration, with an annealing time for a single anneal of $20\mu s$ (D-Wave 2X \cite{tran2016hybrid}). 
 However our algorithm improves the representation of the objective function into the annealer during the search, thus we can conjecture a gain in the number of annealer calls and a consequent speed-up of the optimization.

\section*{Acknowledgements}
 
The present work is supported by:

\vspace{-0.2cm}

\begin{center}
 \includegraphics[width=4cm]{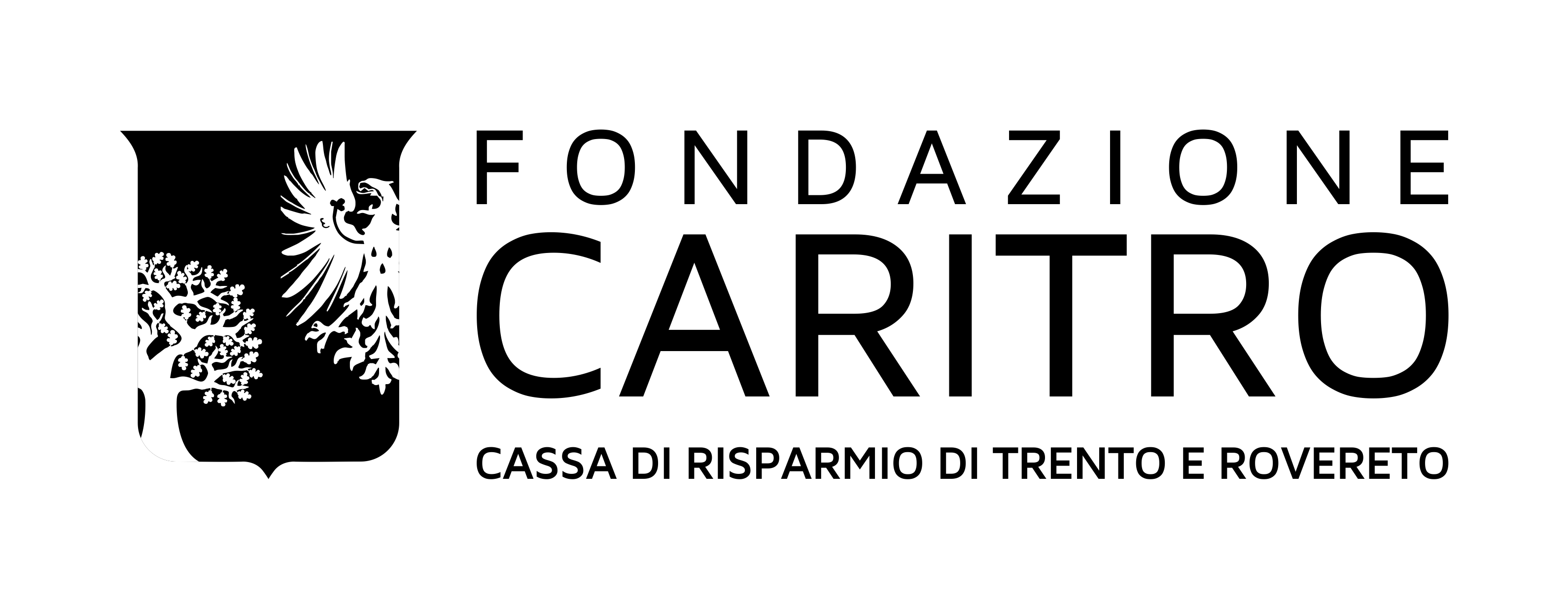}
\end{center}


\end{document}

%% file: algorithm_description.tex
The idea of penalizing some solutions in the annealer initialization can be exploited for defining quantum-classical hybrid schemes for search algorithms with a learning mechanism of the problem encoding into the quantum architecture. Algorithm~\ref{Schema} shows their general form. The goal is to find the minimum of $f$ with a quantum annealer whose adjacency matrix is $\cA$ and its energy function is $\mathsf E(\bTheta, \vz)$.

\vspace{1cm}

 \input{pseudocode_general_schema.tex}\FloatBarrier

\input{pseudocode_QUBO.tex}\FloatBarrier

\noindent
The key point of the learning scheme is the definition and updating of the tabu matrix $S$. Initially two candidate solutions are generated (Algorithm~\ref{Schema}, lines 1-6) and while the best is used for a usual initialization (Algorithm~\ref{Schema}, line 4) the worst is used for initializing the tabu matrix (Algorithm~\ref{Schema}, lines 5-6). In the iterative part of the scheme (Algorithm~\ref{Schema}, lines 7-18) new candidate solutions are repeatedly generated and tested, the candidates that are discarded are used to upgrade the tabu matrix (Algorithm~\ref{Schema}, line 16). 
This permits to generate successive encodings $\mu:\mathcal O\rightarrow\mathcal E$ of the objective function into the energy function, and consequently solutions (Algorithm~\ref{Schema}, lines 8-9), that take into account: 1) the mapping corresponding to the best solution so far and  2) the tabu matrix scaled by a balancing factor $\lambda$. In general $\lambda$ is dynamically updated (Algorithm~\ref{Schema}, line 17) so it can take into account the number of rejected solutions and/or the actual values of the matrix $S$. The balancing factor balances the contributions of the tabu matrix and the objective function (Algorithm~\ref{Schema}, line 9).

\noindent
Algorithm~\ref{Alg} instantiates the general scheme of quantum annealing learning search in the QUBO case, and it takes the structure of a simulated annealing with quantum annealer calls. Differently from Algorithm~\ref{Schema}, the \emph{encodings} are here completely specified by randomly generated permutations. The \emph{candidate solution generation} combines the tabu strategy of Algorithm~\ref{Schema} with a random perturbation of the candidate. The \emph{candidate solution acceptance} allows for random acceptance of suboptimal candidates and the \emph{termination conditions} are explicitly given. The details of Algorithm~\ref{Alg} are presented below. 
\\
\textbf{Encodings.} In Algorithm~\ref{Alg} the tabu-implementing encodings are generated with permutations according to formulas (\ref{encod0}) and (\ref{encod}).   
To this end we call an additional function $g(P,pr)$ (Algorithm~\ref{Alg2}) that modifies a permutation by considering for shuffling each element with probability $p$. The function $g(P,pr)$ produces the permutations that induce the mappings of the objective function into the annealer architecture (Algorithm~\ref{Alg} lines 5 and 23). 
The mappings to the parameters of the annealer take into account the actual annealer topology represented by the graph matrix $\cA$ (Algorithm~\ref{Alg}, lines 6 and 24). 
The permutations are also used to remap the solutions found by the annealer to the initial space of solutions of the problem (Algorithm~\ref{Alg}, lines 8 and 26) and to represent the best map $P^*$.
As an effect of Algorithm~\ref{Alg2}, if $p=1$ the resulting permutation is purely random; with $0<p<1$ the permutation resembles partially the initial one; if $p=0$ the output permutation would coincide with the input one. However, this last circumstance does not occur because periodically  the probability of an element to be shuffled decreases to the value $0<p_\delta<0.5$ with rate $\eta$ (Algorithm~\ref{Alg}, lines 20-22).

\noindent
\textbf{Candidate solution generation.} The matrix $Q$ of the QUBO problem interacts additively with the tabu matrix $S$ scaled by a balancing factor $\lambda$ (Algorithm~\ref{Alg}, line 19) in order to guide the search of the solution by quantum annealing with an energy profile consistent with (\ref{encod}). A crucial effect of this summation is that, in its iterative part, the algorithm does not search anymore for solutions of subproblems as done instead in its initialization. In fact, $S$ contains information about the bad candidates (Algorithm~\ref{Alg}, lines 15 and 32) whose objective function values are greater than $f^*$. 
The balancing factor $\lambda$, initially set to $\lambda_0$, is decreased  (Algorithm~\ref{Alg}, line 37), for example a sensible choice could be $\lambda\leftarrow \min\left\{\lambda_0,\frac{\lambda_0}{2+i-e}\right\}$ where $2+i-e$ is the current number of rejected solutions. 
The candidate solution found by the annealer (Algorithm~\ref{Alg}, line 26) is then perturbed with probability $q$ (Algorithm~\ref{Alg}, line 27) by the function $h(\vz', p)$ (Algorithm~\ref{Alg4}) in order to guarantee the convergence (see Section \ref{Conv}).

\input{pseudocode_matrix_permutation.tex}   \FloatBarrier

\noindent
\textbf{Candidate solution acceptance.} At line 35 of Algorithm~\ref{Alg} a suboptimal solution is accepted with probability $p^{(f'-f^*)}$. By direct comparison with the common rule of acceptation in simulated annealing, namely $p^{(f'-f^*)}=e^{-\frac{(f'-f^*)}{T}}$, it is possible to interpret the role of $p$ by observing that $p=e^{-1/T}$. Namely, in terms of simulated annealing, the parameter $p$ is the probability to accept a suboptimal solution at temperature $T=-{\ln}^{-1}p$ when the objective function worses of a unity.

\input{pseudocode_solution_perturbation.tex}

\noindent
\textbf{Termination conditions.}
The cycle of Algorithm~\ref{Alg} ends with convergence or when the maximum number of iterations is reached.   
Line 17 of Algorithm~\ref{Alg} defines three counters for controlling the end of the cycle: $e$ counts the number of consecutive times that a current solution is generated (Algorithm~\ref{Alg}, line 39); $d$ counts the number of times the current solution and the new solution differ and the current is better (Algorithm~\ref{Alg}, line 34); finally, the variable $i$ simply counts the number of iterations. The counters are compared against input parameters in the termination condition.
\\
Let us illustrate the form that the mappings (\ref{tab}) and (\ref{tabumu}), included in Algorithm~\ref{Schema}, have in Algorithm~\ref{Alg} where they are induced by permutations as in (\ref{encod0}) and (\ref{encod}).
 More precisely, in Algorithm~\ref{Alg} the maps $\mu_1$, $\mu_2$,  introduced in the general scheme (Algorithm~\ref{Schema}, line 1), are: 
\begin{equation}\label{encoding}
	\begin{array}{ll}
	  j=1,2 \hs\hs\hs	\mu_j[f](\vz)=\mathsf E(\vt_j, \pi(\vz)) =\mathsf E(P_j^TQP_j\circ \cA , P_j\vz ).\\
		\end{array}
\end{equation}

\noindent
Moreover, in Algorithm 2 the map $\mu$ (generated at line 8 in the general scheme of Algorithm 1) is:
\begin{equation}\label{encodingq}
\mu[f](\vz)=	\mathsf E(\vt, \pi(\vz))=	\mathsf E(P^T(Q+\lambda S)P\circ \cA, P\vz).
\end{equation}

\noindent
A comparison between (\ref{encoding}) and (\ref{encodingq}) shows that during the initialization the candidates are found as solutions of subproblems whereas the problems solved in the iterative part are not in general subproblems of the full problem. 
\\
Let us remark that in the real QA setting the weights do not range over all $\bR$, a physical machine is characterized by bounded ranges so that: $\theta_i\in[-\delta, +\delta]$ and $\theta_{ij}\in[-\gamma, +\gamma]$, for the D-Wave 2000Q $\delta=2$ and $\gamma=1$ \cite{Roberto}. However the D-Wave system has a built-in scaling function  that
automatically scales the problem to make maximum use of the available ranges \cite{headquarters2018d}, so arbitrary real weights can be used as inputs.

%% file: pseudocode_general_schema.tex
\begin{algorithm}[ht!]
\small
\KwData{Annealer adjacency matrix $\cA$ of order $n$}
\KwIn{$f(\vz)$ to be minimized with respect to $\vz\in\{-1,1\}^n$, annealer energy function $\mathsf E(\bTheta, \vz)$}
\KwResult{$\vz^*$ vector minimum of $f$} 
\SetKwProg{ffun}{function}{:}{}
randomly generate two maps $\mu_1$ and $\mu_2$ of $f$: $\mu_1[f](\vz)$ := $\mathsf E(\bTheta_1[f],\pi_1(\vz))\equiv$ $\equiv\mathsf E(\bTheta_1,\pi_1(\vz))$ and $\mu_2[f](\vz)$ := $\mathsf E(\bTheta_2[f],\pi_2(\vz))\equiv\mathsf E(\bTheta_2,\pi_2(\vz))$
where $\pi_1$ and $\pi_2$ are permutations of the variables\;
find $\vz_1$ and $\vz_2$ whose images $\pi_1(\vz_1)$,   $\pi_2(\vz_2)$ estimate the minima of $\mathsf E(\bTheta_1,\cdot)$  and $\mathsf E(\bTheta_2,\cdot)$ in the annealer\;
evaluate  $f(\vz_1)$ and $f(\vz_2)$\;
use the best to initialize $\vz^*$\  and the map $\mu^*$\;
use the worst to initialize $\vz'$\;
initialize the tabu matrix: $S \gets \vz' \otimes \vz'-I_n+\mbox{diag}(\vz')$; initialize the balancing factor $\lambda$\;
\Repeat{ convergence or maximum number of iterations  conditions } {
from $\mu^*$ and $S$ generate $\vt[f]$ and $\pi$:\\ apply the mapping $\mu[f](\vz)$ := $\mathsf E(\bTheta[f]+\lambda S_{\pi}\circ \cA,\pi(\vz)) \equiv\mathsf E(\bTheta,\pi(\vz))$ with $S_{\pi}=P_\pi^T SP_\pi$ where $P_\pi$ is the matrix of the permutation $\pi$\;
find $\vz'$ whose image $\pi(\vz')$ estimate the minimum of $\mathsf E(\bTheta,\cdot)$ in the annealer\;
  \If{$\vz' \not= \vz^*$} {
  evaluate $f(\vz')$\;
    \If{$\vz'$ is better of $\vz^*$}{$swap(\vz',\vz^*)$;  $\mu^* \gets \mu$\;} 
    use $\vz'$ to update the tabu matrix $S$: $S \gets S+\vz' \otimes \vz'-I_n+diag(\vz')$\; update the balancing factor $\lambda$\;
    }
    }
\Return $\vz^*$\;

\vspace{1.0cm}

\caption{Quantum Annealing Learning Search General Scheme}
\label{Schema}
\end{algorithm}

%% file: pseudocode_QUBO.tex
\begin{algorithm}[h]
\footnotesize

\vspace{-2cm}

\KwData{Matrix $Q$ of order $n$ encoding a QUBO problem,  annealer adjacency matrix $\cA$ of order $n$}
\KwIn{Energy function of the annealer $\mathsf E(\bTheta, \vz)$, permutation modification function $g(P,p)$, minimum probability $0<p_\delta<0.5$ of permutation modification, probability decreasing rate $\eta>0$, candidate perturbation probability $q>0$, number $N$ of iterations at constant $p$, initial balancing factor $\lambda_0>0$, number of annealer runs $k\geq 1$, termination parameters $i_{max}$, $N_{max}$, $d_{min}$ }
\KwResult{$\vz^*$ vector with $n$ elements in $\{-1,1\}$ solution of the QUBO problem} 
\SetKwProg{ffun}{function}{:}{}
\ffun{$f_Q$($x$)}{     \KwRet 
 $x^TQx$ \;}
 $P \gets I_n$\;
 $p \gets 1$\;
 $P_1 \gets g(P,1); P_2 \gets g(P,1)$; \tcp*[h]{\it call Algorithm~\ref{Alg2} for generating two permutation matrices}\\
 $\bTheta_1 \gets P_1^TQP_1\circ \cA$;  $\bTheta_2 \gets P_2^TQP_2\circ \cA$; \tcp*[h]{\it weights initialization according to (\ref{encoding})}\\
 run the annealer $k$ times with weights $\bTheta_1$ and $\bTheta_2$\\
 $\vz_1 \gets P_1^T \widehat{\mbox{argmin}}_\vz(\mathsf E(\bTheta_1,\vz))$;  $\vz_2 \gets P_2^T \widehat{\mbox{argmin}}_\vz(\mathsf E(\bTheta_2,\vz))$; \tcp*[h]{\it estimate energy argmin, $P_1^T$ and $P_2^T$ map back the variables}\\
$ f_1 \gets f_Q(\vz_1); f_2 \gets f_Q(\vz_2)$ ; \tcp*[h]{\it evaluate $f_Q$}\\
\tcp*[h]{\it use the best to initialize $\vz^*$ and $P^*$; use the worst to initialize $\vz'$ }\\
  \eIf{$f_1<f_2$}{$\vz^* \gets \vz_1$; $f^* \gets f_1$; $P^* \gets P_1$ $\vz' \gets \vz_2$\;}{$\vz^* \gets \vz_2$ ; $f^* \gets f_2;$ $P^* \gets P_2$; $\vz' \gets \vz_1$\;}
\lIf( \tcp*[h]{\it use $\vz'$ to initialize the tabu matrix $S$}){$f_1\not =f_2$}{$S \gets \vz' \otimes \vz'-I_n+diag(\vz')$}
\lElse( \tcp*[h]{\it otherwise set all the elements of $S$ to zero}){$S \gets 0$}
$e \gets 0$; $d \gets 0$; $i \gets 0$; $\lambda\gets\lambda_0$\;
\Repeat{$i=i_{max}$ or ($e+d \geq N_{max}$ and $d<d_{min}$) }{
$Q' \gets Q+\lambda S$; \tcp*[h]{\it scale and add the tabu matrix   } \\ 
\If{$N$ divides $i$}{$p \gets p-(p-p_\delta)\eta$;}
$P \gets g(P^*,p)$; \tcp*[h]{\it call Algorithm~\ref{Alg2} that returns a modified permutation}\\
$\bTheta' \gets P^TQ'P\circ \cA$; \tcp*[h]{\it weights initialization according to (\ref{encodingq}) }\\
run the annealer $k$ times with weights $\bTheta'$\\ 
$\vz' \gets P^T \widehat{\mbox{argmin}}_\vz(\mathsf E(\bTheta',\vz))$; \tcp*[h]{\it estimate energy argmin, $P^T$ maps back the variables}\\
with probability $q\hs$ $\vz' \gets h(\vz',p)$; \tcp*[h]{\it possibly perturb the candidate by calling Algorithm~\ref{Alg4}}\\  
 \eIf{$\vz' \not= \vz^*$} {
$f' \gets f_Q(\vz')$; \tcp*[h]{\it evaluate $f_Q$}\\ 
 \eIf{$f'<f^*$}{$swap(\vz',\vz^*)$; $f^* \gets f'$; $P^* \gets P$; $e \gets 0$; $d \gets 0$; \tcp*[h]{\it $\vz'$ is better}\\ 
 $S \gets S+\vz' \otimes \vz'-I_n+diag(\vz')$; \tcp*[h]{\it use $\vz'$ to update the tabu matrix $S$}}
 { $d \gets d+1$\;
  with probability $p^{(f'-f^*)}$ $swap(\vz',\vz^*)$; $f^* \gets f'$; $P^* \gets P$; $e \gets 0$ \; 
 }
   update the balancing factor $\lambda$ with $\lambda\leq\lambda_0$\;
 }
 {$e \gets e+1$\;}
 $i \gets i+1$\;
 }
\Return $\vz^*$\;

\vspace{0.3cm}

\caption{Quantum Annealing Learning Search for QUBO problems.}
\label{Alg}
\end{algorithm}

%% file: pseudocode_matrix_permutation.tex
\begin{algorithm}[h]
\KwIn{A permutation matrix $P$ of order $n$ ($p_i$ row vector of $P$), a probability $pr$ of an element to be considered for shuffling}
\KwResult{A permutation Matrix $P'$}
associative map $m$\;
\ForEach{$i  \gets 1 \dots n$}
 { with probability $pr$ $m[i] \gets i$; \tcp*[h] {\it select the elements to be shuffled}\\
 }
shuffle map $m$\;
\ForEach{$i  \gets 1 \dots n$} 
 {\eIf { $i \in m$} 
       {${p_i}' \gets p_{m[i]}$}
      {${p_i}' \gets p_i$}
  }
return P';
\caption{Permutation Matrix modification function $g(P, pr)$}
\label{Alg2}
\end{algorithm}

%% file: pseudocode_solution_perturbation.tex
\begin{algorithm}[h]
\KwIn{A vector $\vz\in\{-1,1\}^n$, a probability $pr$ of a component to be switched}
\KwResult{A vector in $\{-1,1\}^n$}
\ForEach{$i  \gets 1 \dots n$}
 { with probability $pr$ $z_i \gets -z_i$; \tcp*[h] {\it switch the component}\\
 }
return \vz;
\caption{Candidate perturbation function $h(\vz, pr)$}
\label{Alg4}
\end{algorithm}

%% file: related.tex
\section{Related work}

Tran et al.~\cite{tran2016hybrid} propose a hybrid quantum-classical approach that includes a quantum annealer. In their proposal the classical part of the algorithm maintains a search tree and decides which of two solvers, a quantum annealer and classical one, to call on a subproblem.  Notably they implemented their approach on a D-Wave quantum annealer and their goal is to realize a complete search. Earlier proposals of quantum-classical schemes include the work of Rosenberg et al. \cite{Rosenberg2016} who proposed to decompose a big problem in subproblems that can be solved by the annealer and validated their approach using simulated annealing. More recently Abbott et al. \cite{DBLP:journals/corr/abs-1803-04340} suggested that the search for the embedding of the problem into the annealer topology can be critical in the quest to achieve a quantum speed-up for an hybrid quantum-classical algorithm. The D-Wave handbook \cite{headquarters2018d} reports at page 74 that their \emph{QSage} package implements a hybrid tabu search algorithm for problem decomposition for the annealer. Moreover, D-Wave also released the software \emph{qbsolv} that allows for a decomposition of large QUBO problems.
\\
Our approach differs from what previously published in two fundamental ways: 1) we inserted the tabu mechanism inside the annealer stage by means of the matrix $S$ and 2) we provide a proof of convergence by mapping the classical part to results about tabu/SA procedures.